\documentclass[10pt, draftcls, onecolumn]{IEEEtran}
\usepackage[dvips]{graphicx}
\usepackage{subfigure}
\usepackage{epsfig, psfrag, amsmath, amssymb, amsfonts, latexsym, eucal, balance, indentfirst}
\usepackage{bbm}

\newtheorem{thm}{Theorem}%[section]

\makeatletter

\newcommand{\Rmnum}[1]{\expandafter\@slowromancap\romannumeral #1@}
\makeatother

\graphicspath{{figs/}}

\begin{document}

%%%%%%%%%%%%%%%%%%%%%%%%%%%%%%%%%%%%%%%%%%%%%%%%%%%%%%%%%%%%%%%%%%%%%%%%%%%%%%%%%%%%%%%%%%%%%%%%%%%%%%%%%%%%
\title{Stochastic Analysis of Mean Interference for RTS/CTS Mechanism}
\author{Yi Zhong and Wenyi Zhang\\
Department of Electronic Engineering and Information Science\\
University of Science and Technology of China\\
Hefei 230027, China. Email: {\tt wenyizha@ustc.edu.cn}
}
\maketitle

\begin{abstract}
The RTS/CTS handshake mechanism in WLAN is studied using stochastic geometry. The effect of RTS/CTS is treated as a thinning procedure for a spatially distributed point process that models the potential transceivers in a WLAN, and the resulting concurrent transmission processes are described. Exact formulas for the intensity of the concurrent transmission processes and the mean interference experienced by a typical receiver are established. The analysis yields useful results for understanding how the design parameters of RTS/CTS affect the network interference.
\end{abstract}

\section{Introduction}
\label{sec:intro}
The universal deployment of WLAN during the recent years has been a major driving force for the proliferation of personal wireless network access. In IEEE 802.11 MAC \cite{ieee1997wireless} which is the most common WLAN medium access protocol, both physical carrier sensing and virtual carrier sensing are introduced to avoid the collision in communication. In physical carrier sensing, a potential transmitter first detects whether there are any active transmitters within a detectable region and defers its transmission if so. In virtual carrier sensing, a Request-to-Send/Clear-to-Send (RTS/CTS) handshake mechanism is introduced in order to solve the hidden terminal problem \cite{tobagi1975packet}. In wireless networks, a hidden node is a node which is visible from a receiver, but not from the transmitter communicating with the said receiver. To avoid the collision caused by hidden nodes, the RTS/CTS handshake mechanism sets up a protection zone around a receiver, and thus the overall protection region of a transceiver pair is the union of the carrier sensing cleared region and the RTS/CTS cleared region; see Figure \ref{fig:Pair}.
\begin{figure}[ht]
\centering
\includegraphics[width=0.45\textwidth]{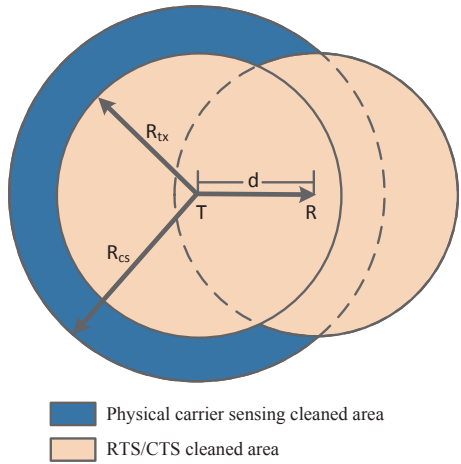}
\caption{Illustration of a pair of transmitter and receiver with the RTS/CTS mechanism.}
\label{fig:Pair}
\end{figure}

Despite of the existing significant body of literature for the performance analysis of 802.11 WLAN (see, e.g., \cite{weinmiller1995analyzing} \cite{cali1998ieee} \cite{bianchi2000performance} \cite{xu2003effectiveness}), the analysis of the RTS/CTS mechanism for a network is still elusive due to the randomness of spatially distributed nodes and the strong correlation among them. Since the purpose of RTS/CTS is to mitigate the network interference, it is important to understand how its design parameters such as the carrier sensing ranges affect the network interference. For applications, a system engineer may be interested to know, under which conditions the network interference is comparable with the thermal noise and thus essentially negligible.

In this paper, we develop a spatial distribution model for nodes under the RTS/CTS mechanism, and evaluate the mean interference that a typical receiver experiences in such a network. Traditional use of the homogeneous Poisson point process (PPP) \cite{haenggi2009stochastic} is not accurate to describe the spatial distribution of transmitters, since with the effect of RTS/CTS there is a strong coupling among neighbor transmitters. Additionally, former interference analysis mainly focuses on the interference at either a typical transmitter or an arbitrary and fixed location, instead of a typical receiver. In our work, we propose a marked point process to characterize the distribution of both transmitters and their corresponding receivers under the RTS/CTS mechanism. In the proposed marked point process, each point indicates the location of a transmitter with a mark denoting the relative location of the associated receiver, and the RTS/CTS mechanism takes effect in a fashion similar to Mat\'{e}rn's hard core models \cite{mat¨¦rn1986spatial}. By employing the model, we follow a general method to derive an exact formula for the mean interference experienced by a typical receiver. A key ingredient in our derivation is to relate the reduced Palm expectation with the second-order factorial measure and the second-order product density of the marked point process, which eventually yield convenient geometric interpretation. With the established results, we reveal how the performance of the RTS/CTS mechanism varies with the density of network nodes and the carrier sensing ranges.

In related works such as \cite{nguyen2007stochastic} and \cite{haenggi2011mean} the analysis is mainly based on Mat\'{e}rn's hard core model, which is not capable of precisely describing the behavior of the CTS/RTS mechanism. Furthermore, in our work the analysis yields exact results for the mean interference, instead of upper/lower bounds or approximations usually adopted in the literature.

The remaining part of this paper is organized as follows. Section \ref{sec:model} describes the marked point process models for nodes under RTS/CTS. Section \ref{sec:analysis} then establishes the main analytical results of this paper, including formulas of the intensity of the marked point process and the mean interference experienced by a typical receiver. Numerical illustration and simulation verification of the analysis will be presented in the forthcoming full version of paper.

\section{Network Model}
\label{sec:model}

In this section, we introduce the statistical network model for pairs of transmitters and receivers, under the RTS/CTS mechanism.

\subsection{Basic Assumptions}
\label{subsec:basic}

We consider a geographic area in which a large number of potential transmitters and their corresponding receivers reside. All the considered transmitters would send their signals in a common frequency band; that is, we base our model on only one typical subchannel in a 802.11 network system. A common assumption is that the subchannels are mutually independently allocated, in which case the analysis of this paper can be directly extended to multi-channel systems.

For simplicity, we set the transmit power of each potential transmitter to a common value $P_t$. Regarding signal propagation, we consider a general deterministic path-loss function $l(\cdot)$, so that the power received at a point of distance $r$ to the transmitter, denoted by $P_r$, is given by $P_r = l(r) P_t$. For example, a simple model commonly used in analysis is $l(r) = A r^{-\alpha}$ where $\alpha > 2$ is the so-called path-loss exponent. In this paper, statistical fading/shadowing and variable transmitting power are not incorporated into the model, and they may be treated by the approach presented herein under the condition that they are independent of the spatial process of nodes.

\subsection{Concurrent Transmission Processes}
\label{subsec:ctp}

To characterize the locations of transmitters and receivers, we begin with a Poisson bipolar model \cite{baccelli2009stochastic} where each point of a Poisson point process is a potential transmitter and has its receiver at some fixed distance $d$ and a random orientation of angle $\theta$. For simplicity, we only consider the case where $d$ is fixed, while the approach may be extended to the case where $d$ is a random variable. By applying the RTS/CTS mechanism, a pair of transceivers sets up an exclusion zone in which other potential transmitters are prohibited to transmit. We assume that the physical carrier sensing cleared region is a circular area centered at the transmitter with radius $R_\mathrm{cs}$, and that the virtual carrier sensing cleared region is the union of two circular areas, which are centered at the transmitter (i.e., RTS) and the receiver (i.e., CTS), respectively, with the same radius $R_\mathrm{tx} < R_\mathrm{cs}$. See Figure \ref{fig:Pair} for illustration.

We call the effect of RTS/CTS as RTS/CTS thinning, which selectively removes some of the points of the Poisson bipolar process of potential transceivers following the CTS/RTS rule. We consider two types of RTS/CTS thinning, similar to Mat\'{e}rn's hard core model. In type \Rmnum{1} thinning, a given transceiver pair is kept only when there is no other transmitter lying in the exclusion zone of the said transceiver pair.
In type \Rmnum{2} thinning, each transceiver pair is endowed with a random mark (as time stamp) and a given transceiver pair is kept only when there is no other transmitter with a smaller mark lying in the exclusion zone of the said transceiver pair. In view of the actual RTS/CTS mechanism, type \Rmnum{2} thinning better captures the reality, while type \Rmnum{1} thinning is overly conservative removing too many potential transceivers.

\subsection{Mathematical Description}
\label{subsec:math-model}

\subsubsection{Type \Rmnum{1} Concurrent Transmission Process}
We begin with $\widetilde{\Phi}_o=\{(X_i,\theta_i,e_i)\}$ which is assumed to be an independently marked Poisson point process (i.m. Poisson p.p.) with intensity $\lambda_p$ on $\mathbb{R}^2$, where
\begin{itemize}
\item
$\Phi_o=\{X_i\}$ denotes the locations of potential transmitters;
\item
$\theta_i$ denotes the orientation of the receiver for node $X_i$, uniformly distributed in $[0,2\pi]$.
Having assumed a constant distance of $d$ between a transmitter and its receiver, the orientation $\theta_i$ together with $X_i$ uniquely determines the location of the receiver.
\item
$e_i$ is the medium access indicator.
\end{itemize}

For $(X_i,\theta_i,e_i)\in\widetilde{\Phi}_o$, let
\begin{eqnarray}
\mathcal{N}(X_i,\theta_i,e_i)=\{(X_j,\theta_j,e_j)\in\widetilde{\Phi}_o: \nonumber\\
X_j\in B_{X_i}(R_{\mathrm{cs}})\cup B_{X_i+(d\cos\theta_i,d\sin\theta_i)}(R_{\mathrm{tx}}),j\neq i\}
\end{eqnarray}
be the set of neighbors of node $(X_i,\theta_i,e_i)$, where $B_X(r)$ denotes the circular area centered at $X$ with radius $r$.

The medium access indicator $e_i$ of node $(X_i,\theta_i,e_i)$ is a dependent mark defined as follows:
\begin{equation}
e_i=\mathbbm{1}\left(\sharp\mathcal{N}(X_i,\theta_i,e_i)=0\right),
\end{equation}
where $\sharp$ denotes the number of elements in its operand set.

Thus, the type \Rmnum{1} RTS/CTS thinning transforms $\widetilde{\Phi}_o$ into
\begin{equation}
\widetilde{\Phi}:=\{(X_i,\theta_i):(X_i,\theta_i,e_i)\in\widetilde{\Phi}_o,e_i=1\},
\end{equation}
which defines the set of transceivers retained by the thinning procedure. Finally, let $\Phi$ consist of the retained transmitters, as follows:
\begin{equation}
\Phi:=\{X_i:(X_i,\theta_i)\in\widetilde{\Phi}\}.
\end{equation}

\subsubsection{Type \Rmnum{2} Concurrent Transmission Process}
The main difference of type \Rmnum{2} process compared with type \Rmnum{1} process is that time is taken into consideration. Since transmission attempts happen asynchronously among nodes, the RTS/CTS mechanism functions in chronological order. We assume $\widetilde{\Phi}_o=\{(X_i,\theta_i,m_i,e_i)\}$ to be an i.m. Poisson p.p. with intensity $\lambda_p$ on $\mathbb{R}^2$, where
\begin{itemize}
\item
$\Phi_o=\{X_i\}$ denotes the locations of potential transmitters;
\item
$\theta_i$ denotes the orientation of the receiver for node $X_i$, uniformly distributed in $[0,2\pi]$;
\item
$\{m_i\}$ are i.i.d. time stamp marks uniformly distributed in $[0,1]$ (with appropriate normalization).
\item
$e_i$ is the medium access indicator.
\end{itemize}

For $(X_i,\theta_i,m_i,e_i)\in\widetilde{\Phi}_o$, let
\begin{eqnarray}
\mathcal{N}(X_i,\theta_i,m_i,e_i)=\{(X_j,\theta_j,m_j,e_j)\in\widetilde{\Phi}_o:\nonumber\\
X_j\in B_{X_i}(R_{\mathrm{cs}})\cup B_{X_i+(d\cos\theta_i,d\sin\theta_i)}(R_{\mathrm{tx}}),j\neq i\}
\end{eqnarray}
be the set of neighbors of node $(X_i,\theta_i,m_i,e_i)$.

The medium access indicator $e_i$ of node $(X_i,\theta_i,m_i,e_i)$ is a dependent mark defined as follows:
\begin{equation}
e_i=\mathbbm{1}\left(\forall {(X_j,\theta_j,m_j,e_j)\in\mathcal{N}(X_i,\theta_i,m_i,e_i)}, m_i<m_j\right).
\end{equation}

Thus, the type \Rmnum{2} RTS/CTS thinning transforms $\widetilde{\Phi}_o$ into
\begin{equation}
\widetilde{\Phi}:=\{(X_i,\theta_i):(X_i,\theta_i,m_i,e_i)\in\widetilde{\Phi}_o,e_i=1\}.
\end{equation}
Finally, let $\Phi$ be defined as follows:
\begin{equation}
\Phi:=\{X_i:(X_i,\theta_i)\in\widetilde{\Phi}\}.
\end{equation}

\section{Intensity of Transmitters and Mean Interference}
\label{sec:analysis}

In this section, we establish for both types of concurrent transmission processes the nodes intensity and the mean interference experienced by a typical receiver.

Before presenting our results, it is convenient to introduce the following quantities:
\begin{itemize}
\item $V_o$: the area of the exclusion zone of a transceiver pair, given by (see Figure \ref{fig:Pair})
\begin{equation}
V_o=(\pi-\gamma_1)R_{\mathrm{cs}}^2+(\pi-\gamma_2)R_{\mathrm{tx}}^2+dR_{\mathrm{cs}}\sin\gamma_1, \label{equ:V_o}
\end{equation}
where $\gamma_1=\arccos\Big(\frac{d^2+R_{\mathrm{cs}}^2-R_{\mathrm{tx}}^2}{2dR_{\mathrm{cs}}}\Big)$ and $\gamma_2=\arccos\Big(\frac{d^2+R_{\mathrm{tx}}^2-R_{\mathrm{cs}}^2}{2dR_{\mathrm{tx}}}\Big)$.

\item $S_1=\{(r,\varphi,\theta):r\leq R_{\mathrm{cs}}\}$, $S_2=\{(r,\varphi,\theta):r^2-2rd\cos\varphi+d^2\leq R_{\mathrm{tx}}^2\}$, $S_3=\{(r,\varphi,\theta):r^2+2rd\cos(\varphi-\theta)+d^2\leq R_{\mathrm{tx}}^2\}$. Consider two transmitters of distance $r$ apart and with phase angle difference $\varphi$, the orientation marks of which are $0$ and $\theta$, respectively; see Figure \ref{fig:twopair}.
    Then, $S_1$ denotes the event that the two transmitters are within the physical sensing cleared region of each other, $S_2$ (and $S_3$) denotes the event that one transmitter is within the RTS/CTS cleared region of the other.

\item $V(r,\varphi,\theta)$: the area of the union of two transceiver pairs, whose transmitters are distance $r$ apart, with relative phase difference $\varphi$, and whose receivers are at orientations $0$ and $\theta$ respectively; see Figure \ref{fig:twopair}.
\begin{figure}
\centering
\includegraphics[width=0.4\textwidth]{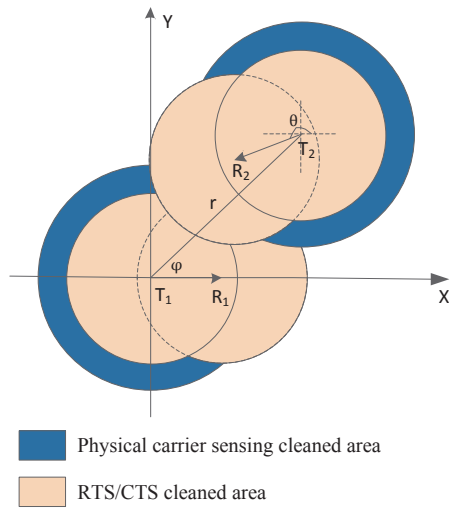}
\caption{Union of two transceiver pairs.}
\label{fig:twopair}
\end{figure}
\end{itemize}

\subsection{Type \Rmnum{1} Concurrent Transmission Process}
\subsubsection{Nodes Intensity}

To characterize the nodes intensity, we apply the notion of Palm probability. Consider a stationary point process $\Phi$ with finite non-zero intensity $\lambda$. Note that $\Phi$ is a random variable taking values in the measurable space $(\mathbb{M},\mathcal{M},P)$, where $\mathbb{M}$ is the set of all possible simple point process, $\mathcal{M}$ is a $\sigma-$algebra defined on $\mathbb{M}$ and $P$ is the probability measure. The Palm distribution induced by $\Phi$ is
\begin{eqnarray}
P_o(Y) = \frac{1}{\lambda v_d(B)}\int\sum_{x\in\Phi\bigcap B}\mathbbm{1}_Y(\Phi_{-x})P(\mathrm{d}\Phi) \label{equ:palm}
\end{eqnarray}
where $B$ is an arbitrary Borel set of positive volume, $v_d$ is the Lebesgue measure and $\Phi_{-x}=\{x_i-x|x_i\in\Phi\}$.

The intensity of a stationary point process is defined as the ratio of the expected number of points in a given area to the Lebesgue measure of that area. Considering the type \Rmnum{1} concurrent transmission process, the retained $\Phi$ is a dependent thinning of the original Poisson p.p. $\Phi_o$. Let $Y$ be the event that the point at the origin $o$ is retained, denoted by $Y=\{\Phi_o\in\mathbb{M}:o\in\Phi\}$.
The Palm probability $P_o(Y)$ is thus the probability that the point at $o$ is retained conditioned upon that there is a point of $\Phi_o$ located at $o$. Therefore the intensity of $\Phi$ is given as follows:
\begin{eqnarray}
\lambda&=&\frac{1}{v_d(B)}\int\sum_{x\in\Phi_o\bigcap B}\mathbbm{1}_{\{\Phi_o:x\in\Phi\}}(\Phi_o)P(\mathrm{d}\Phi_o)\nonumber\\
&=&\frac{1}{v_d(B)}\int\sum_{x\in\Phi_o\bigcap B}\mathbbm{1}_Y((\Phi_o)_{-x})P(\mathrm{d}\Phi_o) =\lambda_pP_o(Y)\nonumber\\
&=&\lambda_pP\big((o,\theta_o)\in\widetilde{\Phi}|(o,\theta_o)\big) = \lambda_pe^{-\lambda_pV_o}. \label{equ:lambda_1}
\end{eqnarray}
From (\ref{equ:lambda_1}), we notice that if $\lambda_p$ is small, $\lambda$ increases with $\lambda_p$, and when $\lambda_p$ is large enough, $\lambda$ starts to decrease. In the extreme case as $\lambda_p\rightarrow\infty$, $\lambda\rightarrow 0$. This shows that the type \Rmnum{1} RTS/CTS thinning is overly conservative removing too many potential transceivers. The maximum of $\lambda$ occurs when we choose $\lambda_p = 1/V_o$.

\subsubsection{Mean Interference}

The mean interference of a typical receiver in type \Rmnum{1} concurrent transmission process is given by the following theorem.
\begin{thm}
\label{thm:inter_1}
The mean interference experienced by a typical receiver in type \Rmnum{1} concurrent transmission process is
\begin{eqnarray}
&&E_{(o,0)}^!(I_{z_o}) = \frac{\lambda_p^2P_t}{2\pi\lambda}\int_0^\infty\!\!\!\int_0^{2\pi}\!\!\!\int_0^{2\pi}\nonumber\\
&&l(\sqrt{r^2-2rd\cos\varphi+d^2})k(r,\varphi,0,\theta)r\mathrm{d}\theta\mathrm{d}\varphi\mathrm{d}r \label{equ:inter_final}
\end{eqnarray}
where $k(r,\varphi,0,\theta)$ is given by
\begin{equation}
k(r,\varphi,0,\theta) = \left\{\begin{array}{ll}
0 & S_1\bigcup S_2\bigcup S_3 \\
\exp(-\lambda_pV(r,\varphi,\theta)) & \mathrm{otherwise}.
\end{array} \right.\label{equ:k_final}
\end{equation}
\end{thm}
\begin{proof}
Without loss of generality, we analyze a typical transceiver pair whose transmitter is located at the origin $o$ and the receiver is located at the orientation $\theta_o$; in other words, we let $(o,\theta_o)\in\widetilde{\Phi}$ and the coordinate of the corresponding receiver is $z_o=\{d\cos\theta_o,d\sin\theta_o\}$. The interference experienced by the receiver located at $z_o$ is given by
\begin{eqnarray}
&&E_{(o,\theta_o)}^!(I_{z_o}) = E_{(o,\theta_o)}^!\Big(\sum_{(x,\theta)\in\widetilde{\Phi}}P_tl(|x-z_o|)\Big) \nonumber \\
&=&\int_\mathbb{M}\Big(\int_{R^2\times[0,2\pi]}P_tl(|x-z_o|)\widetilde{\Phi}(\mathrm{d}(x,\theta))\Big)P_{(o,\theta_o)}^!(\mathrm{d}\widetilde{\Phi}) \nonumber \\
&=&\int_{R^2\times[0,2\pi]}P_tl(|x-z_o|)\int_\mathbb{M}\Big(\widetilde{\Phi}(\mathrm{d}(x,\theta))P_{(o,\theta_o)}^!(\mathrm{d}\widetilde{\Phi})\Big) \nonumber \\
&=&\lambda P_t\int_{R^2\times[0,2\pi]}l(|x-z_o|)\mathcal{K}_{\theta_o}(\mathrm{d}(x,\theta)). \label{equ:mean_I}
\end{eqnarray}
Similar to the case of unmarked p.p., we define $\mathcal{K}_{\theta_o}(B\times L), B\subset R^2, L\subset [0,2\pi]$, as the reduced second-order factorial measure of marked p.p. on $B$. Intuitively, $\lambda\mathcal{K}_{\theta_o}(B\times L)$ is the expected number of points located in $B$ with marks taking values in $L$ under the condition that $(o,\theta_o)\in\widetilde{\Phi}$, and its mathematical description is given by
\begin{eqnarray}
\mathcal{K}_{\theta_o}(B\times L)&=&\frac{1}{\lambda}E_{(o,\theta_o)}^!\Big(\widetilde{\Phi}(B\times L)\Big) \nonumber \\
&=&\frac{1}{\lambda}\int_\mathbb{M}\widetilde{\Phi}(B\times L)P_{(o,\theta_o)}^!(\mathrm{d}\widetilde{\Phi}).
\end{eqnarray}
Next we thus focus on the evaluation of $\mathcal{K}_{\theta_o}(B\times L)$.

The second-order factorial measure of the marked p.p. $\widetilde{\Phi}$ is given by \cite[pp. 114]{stoyanstochastic}
\begin{eqnarray}
&&\alpha^{(2)}(B_1\times B_2\times L_1\times L_2)\nonumber\\
&=&\!\!\!E\Bigg(\sum_{\substack{(x_1,\theta_1)\in\widetilde{\Phi} \\ (x_2,\theta_2)\in\widetilde{\Phi}}} ^{\neq}\mathbbm{1}_{B_1}(x_1)\mathbbm{1}_{B_2}(x_2)\mathbbm{1}_{L_1}(\theta_1)\mathbbm{1}_{L_2}(\theta_2)\Bigg) \nonumber \\
&=&\!\!\!\int_\mathbb{M}\!\!\sum_{(x,\theta)\in\widetilde{\Phi}}\!\!\mathbbm{1}_{B_1}(x)\mathbbm{1}_{L_1}(\theta)\widetilde{\Phi}(B_2\times L_2/\{(x,\theta)\})P(\mathrm{d}\widetilde{\Phi}).\nonumber\\
&& \label{equ:alpha1}
\end{eqnarray}
According to the refined Campbell theorem, we have
\begin{eqnarray}
&&E\Big(\sum_{(x,\theta)\in\widetilde{\Phi}}h(x,\theta,\widetilde{\Phi}/\{x,\theta\})\Big)\nonumber\\
&=&\frac{\lambda}{2\pi}\int_{R^2}\int_\Theta\int_\mathbb{M}h(x,\theta,\widetilde{\Phi})P_{(x,\theta)}^!(\mathrm{d}\widetilde{\Phi})\mathrm{d}\theta\mathrm{d}x\nonumber.
\end{eqnarray}
Letting $h\big(x,\theta,\widetilde{\Phi}\big)=\mathbbm{1}_{B_1}(x)\mathbbm{1}_{L_1}(\theta)\widetilde{\Phi}(B_2\times L_2)$ and plugging into (\ref{equ:alpha1}), we get
\begin{eqnarray}
&&\alpha^{(2)}(B_1\times B_2\times L_1\times L_2) \nonumber \\
&=&\int_\mathbb{M}\sum_{(x,\theta)\in\widetilde{\Phi}}h\big(x,\theta,\widetilde{\Phi}/\{(x,\theta)\}\big)P(\mathrm{d}\widetilde{\Phi}) \nonumber \\
&=&\frac{\lambda}{2\pi}\int_{R^2\times[0,2\pi]}\int_\mathbb{M}h\big(x,\theta,\widetilde{\Phi}\big)P_{(x,\theta)}^!(\mathrm{d}\widetilde{\Phi})\mathrm{d}(x,\theta) \nonumber \\
%&=&\lambda\int_{R^2\times[0,2\pi]}\int_\mathbb{M}\mathbbm{1}_{B_1}(x)\mathbbm{1}_{L_1}(\theta)\widetilde{\Phi}_{-x}(B_2\times L_2/\{(x,\theta)\}) P_{(o,\theta)}(\mathrm{d}\Phi)\mathrm{d}(x,\theta) \nonumber \\
&=&\frac{\lambda}{2\pi}\int_{R^2\times[0,2\pi]}\int_\mathbb{M}\mathbbm{1}_{B_1}(x)\mathbbm{1}_{L_1}(\theta)\nonumber\\
&&\quad \widetilde{\Phi}((B_2-x)\times L_2) P_{(o,\theta)}^!(\mathrm{d}\widetilde{\Phi})\mathrm{d}(x,\theta) \nonumber \\
&=&\frac{\lambda}{2\pi}\int_{R^2\times[0,2\pi]}\mathbbm{1}_{B_1}(x)\mathbbm{1}_{L_1}(\theta)\nonumber\\
&&\quad \Bigg(\int_\mathbb{M}\widetilde{\Phi}((B_2-x)\times L_2) P_{(o,\theta)}^!(\mathrm{d}\widetilde{\Phi})\Bigg)\mathrm{d}(x,\theta) \nonumber \\
%&=&\frac{\lambda^2}{2\pi}\int_{R^2\times[0,2\pi]}\mathbbm{1}_{B_1}(x)\mathbbm{1}_{L_1}(\theta)\nonumber\\
%&&\quad \mathcal{K}_{\theta}((B_2-x)\times L_2)  \mathrm{d}(x,\theta) \nonumber \\
&=&\frac{\lambda^2}{2\pi}\int_{B_1\times L_1}\mathcal{K}_{\theta}((B_2-x)\times L_2)  \mathrm{d}(x,\theta). \label{equ:alpha2}
\end{eqnarray}

Consider the second-order product density $\varrho^{(2)}$ defined as follows \cite[pp. 111]{stoyanstochastic}
\begin{eqnarray}
&&\alpha^{(2)}(B_1\times B_2\times L_1\times L_2) \nonumber \\
&=&\int_{B_1\times B_2\times L_1\times L_2}\varrho^{(2)}(x_1,x_2,\theta_1,\theta_2)\mathrm{d}(x_1,x_2,\theta_1,\theta_2) \nonumber \\
&\stackrel{(a)}{=}&\int_{B_1\times L_1}\Bigg(\int_{B_2\times L_2}\varrho^{(2)}(x_2-x_1,\theta_1,\theta_2)\mathrm{d}(x_2,\theta_2)\Bigg)\mathrm{d}(x_1,\theta_1) \nonumber \\
&=&\int_{B_1\times L_1}\Bigg(\int_{(B_2-x_1)\times L_2}\varrho^{(2)}(x_2,\theta_1,\theta_2)\mathrm{d}(x_2,\theta_2)\Bigg)\mathrm{d}(x_1,\theta_1)\nonumber\\
&&\label{equ:alpha3}
\end{eqnarray}
where $(a)$ follows from the fact that for a stationary p.p., $\varrho^{(2)}(x_1,x_2,\theta_1,\theta_2)$ depends only on $x_2-x_1$, $\theta_1$ and $\theta_2$. By comparing (\ref{equ:alpha2}) and (\ref{equ:alpha3}), we get
\begin{eqnarray}
\mathcal{K}_{\theta_o}(B\times L)&=&\frac{2\pi}{\lambda^2}\int_{B\times L}\varrho^{(2)}(x,\theta_o,\theta)\mathrm{d}(x,\theta),
\end{eqnarray}
whose differential form is
\begin{eqnarray}
\mathcal{K}_{\theta_o}(\mathrm{d}(x,\theta))&=&\frac{2\pi}{\lambda^2}\varrho^{(2)}(x,\theta_o,\theta)\mathrm{d}(x,\theta). \label{equ:derive}
\end{eqnarray}
Plugging (\ref{equ:derive}) into (\ref{equ:mean_I}), we get
\begin{eqnarray}
&&E_{(o,\theta_o)}^!(I_{z_o}) =\nonumber\\
&&\frac{2\pi P_t}{\lambda}\!\!\int_{R^2\times[0,2\pi]} \!\!l(|x-z_o|)\varrho^{(2)}(x,\theta_o,\theta)\mathrm{d}(x,\theta). \label{equ:mean_I_2}
\end{eqnarray}

At this point, for convenience and without loss of generality, we assume $\theta_o=0$; that is, the considered typical receiver is located at coordinate $z_o=(d,0)$. The interference experienced at $z_o$ is given by
\begin{equation}
E_{(o,0)}^!(I_{z_o}) = \frac{\lambda_p^2P_t}{2\pi\lambda}\int_{R^2\times[0,2\pi]}l(|x-z_o|)k(x,0,\theta)\mathrm{d}(x,\theta),
\end{equation}
which, when written in polar form, is
\begin{eqnarray}
&&E_{(o,0)}^!(I_{z_o}) = \frac{\lambda_p^2P_t}{2\pi\lambda}\int_0^\infty\!\!\!\int_0^{2\pi}\!\!\!\int_0^{2\pi} \nonumber\\
&&l(\sqrt{r^2-2rd\cos\varphi+d^2})k(r,\varphi,0,\theta)r\mathrm{d}\theta\mathrm{d}\varphi\mathrm{d}r. \label{equ:inter_final_proof}
\end{eqnarray}
The function $k(r,\varphi,0,\theta)$ denotes the probability that two transmitters of distance $r$ apart and with phase angle difference $\varphi$, the orientation marks of which are $0$ and $\theta$, respectively, are both retained (see Figure \ref{fig:twopair}). The value of $k(r,\varphi,0,\theta)$ is then obtained from their geometric relationship and is given by (\ref{equ:k_final}): when any of $S_1$, $S_2$ and $S_3$ occurs, the two transceiver pairs are within their guard zone of the RTS/CTS mechanism and thus neither of them is retained; otherwise, these two pairs are retained if and only if there is no other potential transmitter within the area of $V(r, \varphi, \theta)$, which occurs with probability $\exp(-\lambda_p V(r, \varphi, \theta))$ since $\Phi_o$ is a homogenous Poisson p.p. with intensity $\lambda_p$.
\end{proof}

\subsection{Type \Rmnum{2} Concurrent Transmission Process}
\subsubsection{Nodes Intensity}
In type \Rmnum{2} concurrent transmission process, a transceiver pair with time stamp mark $t$ leads to the removal of those transmitters lying in its induced guard zone and having time stamp mark larger than $t$. Conditioned upon $t$, the transceiver pairs with time stamp mark larger than $t$ is an independent thinning of the original Poisson p.p. $\Phi_o$. According to (\ref{equ:lambda_1}), given $t$, the probability of retaining a given transceiver pair with time stamp mark $t$ is $e^{-\lambda_ptV_o}$. By averaging over $t$, we obtain the intensity of type \Rmnum{2} concurrent transmission process:
\begin{eqnarray}
\lambda&=&\lambda_pP\big((o,\theta_o)\in\widetilde{\Phi}|(o,\theta_o)\big)\nonumber\\
&=&\lambda_p\int_0^1P\big((o,\theta_o)\in\widetilde{\Phi}|(o,\theta_o,m_o)\in\widetilde{\Phi}_o,m_o=t\big)\mathrm{d}t\nonumber\\
&=&\lambda_p\int_0^1e^{-\lambda_ptV_o}\mathrm{d}t = \frac{1}{V_o}(1-e^{-\lambda_pV_o}) \label{equ:lambda_2}
\end{eqnarray}
where $V_o$ is given by (\ref{equ:V_o}). From (\ref{equ:lambda_2}) we notice that $\lambda$ monotonically increasing with $\lambda_p$, and as $\lambda_p\rightarrow \infty$, $\lambda$ tends toward a constant limit $1/V_o$.

\subsubsection{Mean Interference}

In parallel with Theorem \ref{thm:inter_1}, the following theorem gives the mean interference for type \Rmnum{2} concurrent transmission process.
\begin{thm}
\label{thm:inter_2}
The mean interference experienced by a typical receiver in type \Rmnum{2} concurrent transmission process is
\begin{eqnarray}
&&E_{(o,0)}^!(I_{z_o}) = \frac{\lambda_p^2P_t}{2\pi\lambda}\int_0^\infty\!\!\!\int_0^{2\pi}\!\!\!\int_0^{2\pi}\nonumber\\
&&l(\sqrt{r^2-2rd\cos\varphi+d^2})k(r,\varphi,0,\theta)r\mathrm{d}\theta\mathrm{d}\varphi\mathrm{d}r \label{equ:inter_final2}
\end{eqnarray}
where $k(r,\varphi,0,\theta)$ is given by
\begin{equation}
k(r,\varphi,0,\theta) = \left\{\begin{array}{ll}
0 & S_1\bigcup (S_2\bigcap S_3) \\
2\beta(V) & \overline{S_1}\bigcap\overline{S_2}\bigcap\overline{S_3}  \\
\beta(V) & \mathrm{otherwise}.
\end{array} \right.\label{equ:k_final2}
\end{equation}
in which $V$ is the abbreviation of $V(r,\varphi,\theta)$, and $\beta(V)$ is given by
\begin{eqnarray}
\beta(V) = \frac{V_oe^{-\lambda_pV}-Ve^{-\lambda_pV_o}+V-V_o}{\lambda_p^2(V-V_o)VV_o}.
\end{eqnarray}
\end{thm}
\begin{proof}
The derivation of (\ref{equ:inter_final2}) follows the same line as that of the type \Rmnum{1} process.
However, the value of $k(r,\varphi,0,\theta)$, which is the probability that two transceiver pairs are both retained, is different from the type \Rmnum{1} process since the time stamp marks are taken into consideration. We should consider not only the geometric relationship (see Figure \ref{fig:twopair}), but also the relationship between the time stamp marks of the two transceiver pairs. When $S_1\bigcup (S_2\bigcap S_3)$ occurs, at least one of the considered two transceiver pairs has to be removed after comparing their time stamp marks and thus $k(r,\varphi,0,\theta)$ is zero. When $\overline{S_1}\bigcap\overline{S_2}\bigcap\overline{S_3}$, there is no direct comparison between the time stamp marks of the considered two transceiver pairs. Let their time stamp marks be $t_1$ and $t_2$ respectively. So by separately considering the cases of $t_1 \geq t_2$ and $t_1 < t_2$, we can evaluate the probability that both of the transceiver pairs are retained, namely, $k(r,\varphi,0,\theta)$, as
\begin{eqnarray}
\label{eqn:k-two-conditions}
&&k(r,\varphi,0,\theta) = \int_0^1 e^{-\lambda_p t_1 V_o} \left[\int_0^{t_1} e^{-\lambda_p t_2 (V - V_o)} \mathrm{d}t_2\right] \mathrm{d}t_1 \nonumber\\
&&+ \int_0^1 e^{-\lambda_p t_2 V_o} \left[\int_0^{t_2} e^{-\lambda_p t_1 (V - V_o)} \mathrm{d}t_1\right] \mathrm{d}t_2\nonumber\\
&=& 2 \int_0^1 e^{-\lambda_p t_1 V_o} \left[\int_0^{t_1} e^{-\lambda_p t_2 (V - V_o)} \mathrm{d}t_2\right] \mathrm{d}t_1,
\end{eqnarray}
which can be evaluated as twice of $\beta(V)$. Here, $\int_0^{t_1} e^{-\lambda_p t_2 (V - V_o)} \mathrm{d}t_2$ is the conditional probability that, when the first transceiver pair is marked by $t_1$, there is no other transceiver pair of mark smaller than $t_2 < t_1$ lying in the area covered by the second transceiver pair only; $\int_0^{t_2} e^{-\lambda_p t_1 (V - V_o)} \mathrm{d}t_1$ is analogously interpreted. When either $\overline{S_1}\bigcap{S_2}\bigcap\overline{S_3}$ or $\overline{S_1}\bigcap\overline{S_2}\bigcap{S_3}$ occurs, exactly one of the considered transmitter is within the RTS/CTS cleared region of the other, and hence only one of the terms in (\ref{eqn:k-two-conditions}) should be taken into account, resulting into $k(r,\varphi,0,\theta) = \beta(V)$.
\end{proof}

\section{Conclusion}
\label{sec:conclusion}

In this paper, we proposed marked point process models to characterize the spatial distribution of transceivers under the RTS/CTS handshake mechanism in WLAN, focusing on the evaluation of the mean interference experienced by a typical receiver. Our analysis revealed how the mean interference under RTS/CTS varies with system parameters such as the carrier ranges and the density of transceiver nodes. The method employed in our work may be applicable to more general wireless networks modeled by more sophisticated kinds of point processes.

%%%%%%%%%%%%%%%%%%%%%%%%%%%%%%%%%%%%%%%%%%%%%%%%%%%%%%%%%%%%%%%%%%%%%%%%%%%%%%%%%%%%%%%%%%%%%%%%%%%%%%%%%%%%
\bibliographystyle{IEEEtran}
\bibliography{csma-isit}

\end{document}